\documentclass[svgnames,dvipsnames]{llncs}
\usepackage{microtype}
\usepackage{multirow}

\usepackage{amsmath,amssymb}
\newtheorem{observation}{Observation}
\newtheorem{Reduction Rule}{Reduction Rule}

\usepackage{enumerate}

\usepackage{xspace}

\newcommand{\F}{\ensuremath{\mathbb F}}
\newcommand{\G}{\ensuremath{\mathbb G}}

\newcommand{\ICpartition}{IC-partition}
\newcommand{\ICpartitions}{IC-partitions}

\newcommand{\Oh}{\mathcal{O}}

\usepackage[usenames,dvipsnames]{color}

\newcommand{\abpartization}{{\sc Vertex $(r,\ell)$-Partization}}
\newcommand{\twoonepartization}{{\sc Vertex $(2,1)$-Partization}}
\newcommand{\onetwopartization}{{\sc Vertex $(1,2)$-Partization}}
\newcommand{\twotwopartization}{{\sc Vertex $(2,2)$-Partization}}
\newcommand{\abedgepartization}{{\sc Edge $(r,\ell)$-Partization}}
\newcommand{\twooneedgepartization}{{\sc Edge $(2,1)$-Partization}}
\newcommand{\onetwoedgepartization}{{\sc Edge $(1,2)$-Partization}}

\newcommand{\twotwoDeletion}{$(2,2)$-{vertex deletion set}}
\newcommand{\abGraph}{$(r,\ell)$-{graph}}
\newcommand{\onetwoGraph}{$(1,2)$-{graph}}
\newcommand{\twooneGraph}{$(2,1)$-{graph}}
\newcommand{\twotwoGraph}{$(2,2)$-{graph}}
\newcommand{\abPartition}{$(r,\ell)$-{partition}}

\newcommand{\twoonePartition}{$(2,1)$-{partition}}
\newcommand{\twotwoPartition}{$(2,2)$-{partition}}

\newcommand{\absplitGraph}{$(r,\ell)$-{split graph}}

\newcommand{\twotwosplitGraph}{$(2,2)$-{split graph}}
\newcommand{\absplitPartition}{$(r,\ell)$-{split partition}}

\newcommand{\twotwosplitPartition}{$(2,2)$-{split partition}}

\newcommand{\YES}{\textsc{YES}}
\newcommand{\NO}{\textsc{NO}}
\newcommand{\OCT}{\textsc{Odd Cycle Transversal}}
\newcommand{\OCET}{\textsc{Edge Odd Cycle Transversal}}
\newcommand{\TOCT}{\textsc{TOCT}}
\usepackage{amsfonts}
\usepackage{amstext}
\usepackage{graphicx,tikz}
\RequirePackage{fancyhdr}
\usepackage{xcolor}
 \usepackage{boxedminipage}
\newcommand{\defparproblem}[4]{
  \vspace{1mm}
\noindent\fbox{
  \begin{minipage}{0.96\textwidth}
  \begin{tabular*}{\textwidth}{@{\extracolsep{\fill}}lr} #1  & {\bf{Parameter:}} #3
\\ \end{tabular*}
  {\bf{Input:}} #2  \\
  {\bf{Question:}} #4
  \end{minipage}
  }
  \vspace{1mm}
}

\newcommand{\defparproblemoutput}[4]{
  \vspace{1mm}
\noindent\fbox{
  \begin{minipage}{0.96\textwidth}
  \begin{tabular*}{\textwidth}{@{\extracolsep{\fill}}lr} #1  & {\bf{Parameter:}} #3
\\ \end{tabular*}
  {\bf{Input:}} #2  \\
  {\bf{Output:}} #4
  \end{minipage}
  }
  \vspace{1mm}
}

\newcommand{\defproblem}[3]{
  \vspace{1mm}
\noindent\fbox{
  \begin{minipage}{0.96\textwidth}
  \begin{tabular*}{\textwidth}{@{\extracolsep{\fill}}lr} #1 \\ \end{tabular*}
  {\bf{Input:}} #2  \\
  {\bf{Question:}} #3
  \end{minipage}
  }
  \vspace{1mm}
}

\newcommand{\NP}{\text{\normalfont  NP}}

\newcommand{\FPT}{\text{\normalfont FPT}}

\title{Parameterized Algorithms for Deletion to $(r,\ell)$-graphs}

\titlerunning{Parameterized Algorithms for Deletion to $(r,\ell)$-graphs}

\author{Sudeshna Kolay \inst{1} \and Fahad Panolan \inst{1} }

\institute{Institute of Mathematical Sciences, Chennai, India.  
\email{\{skolay|fahad\}@imsc.res.in}
}

\begin{document}
\maketitle

\begin{abstract}
For  fixed integers $r,\ell \geq 0$, a graph $G$ is called an {\em $(r,\ell)$-graph} if the vertex set $V(G)$ can be partitioned into $r$ independent sets and $\ell$ cliques. 
This brings us to the following natural parameterized questions: {\sc Vertex $(r,\ell)$-Partization} and {\sc Edge $(r,\ell)$-Partization}. An input to these problems consist of  a graph $G$ and a positive integer $k$  
and the objective is to decide whether there exists a set $S\subseteq V(G)$ ($S\subseteq E(G)$) such that the deletion of $S$ from $G$ results in an $(r,\ell)$-graph.  These problems generalize well studied problems such as 
{\sc Odd Cycle Transversal}, {\sc Edge Odd Cycle Transversal}, {\sc Split Vertex Deletion} and {\sc Split Edge Deletion}. We do not hope to get parameterized algorithms for either {\sc Vertex $(r,\ell)$-Partization} or {\sc Edge $(r,\ell)$-Partization} when either of $r$ or $\ell$ is at least $3$ as the recognition problem itself is NP-complete. This leaves the case of $r,\ell \in  \{1,2\}$.  We almost complete the parameterized complexity dichotomy for these problems by obtaining the following results: 

\begin{enumerate}
\item We show that {\sc Vertex $(r,\ell)$-Partization} is fixed parameter tractable (\FPT{}) for  
$r,\ell \in  \{1,2\}$.  
Then we design an 
 $\Oh(\sqrt{ \log{n} })$-factor approximation algorithms for these problems. These approximation algorithms are then utilized to design polynomial sized randomized Turing kernels for these problems. 

\item   {\sc Edge $(r,\ell)$-Partization} is \FPT{} when  
$(r,\ell)\in\{(1,2),(2,1)\}$. However, the parameterized complexity of   {\sc Edge $(2,2)$-Partization} remains open. 
\end{enumerate}
For our approximation algorithms and thus for Turing kernels we use 
an interesting finite forbidden induced graph characterization, for a class of graphs known as $(r,\ell)$-split graphs, properly containing the class of  $(r,\ell)$-graphs. This approach to obtain approximation algorithms could be of an independent interest. 
\end{abstract}

\section{Introduction}

For  fixed integers $r,\ell \geq 0$, a graph $G$ is called an {\em \abGraph} if the vertex set $V(G)$ can be partitioned into $r$ independent sets and $\ell$ cliques. Although the problem has an abstract setting, some special cases are well known graph classes and have been widely studied. For example, $(2,0)$- and  $(1,1)$-graphs correspond to bipartite graphs and split graphs respectively.  
A $(3,0)$-graph is a $3$-colourable graph. Already, we get a hint of an interesting dichotomy for this graph class, even with respect to recognition algorithms. Throughout the paper we will use $m$ and $n$ to denote the number of edges and the number of vertices, respectively, in the input graph $G$.  It is well known that we can recognize $(2,0)$- and  $(1,1)$-graphs  in $\Oh(m+n)$ time. In fact, one can show that recognizing whether a graph $G$ is an \abGraph, when $r,\ell \leq 2$, can be done in polynomial time \cite{Brandstadt98thecomplexity,Feder03listpartitions}. On the other hand, when either $r \geq 3$ or $\ell \geq 3$, the recognition problem is as hard as the celebrated $3$-colouring problem, which is \NP-complete~\cite{Garey:1979:CIG:578533}. These problems are also studied when the input is restricted to be a chordal graph, in which case we can get polynomial time recognition algorithms for every $r$ and $\ell$~\cite{DBLP:journals/endm/FederHR11}.

The topic of this paper is to design recognition algorithms for  {\em almost}  \abGraph{s} in the realm of parameterized algorithms. In particular, we study the following natural parameterized questions on \abGraph{s}: \abpartization\ and \abedgepartization. 

\smallskip

 \defparproblem{{\abpartization}}{A Graph $G$ and a positive integer $k$}{$k$}{Is there a vertex subset $S\subseteq V(G)$ of size at most $k$ such that  
 $G-S$ is an \abGraph?} 

  \defparproblem{{\abedgepartization}}{A Graph $G$ and a positive integer $k$}{$k$}{Is there an edge subset $F\subseteq E(G)$ of size at most $k$ such that $G-F$ is an \abGraph?}

These problems generalize some of the most well studied problems in parameterized complexity, such as 
{\sc Vertex Cover}, 
{\sc Odd Cycle Transversal (OCT)}, {\sc Edge Odd Cycle Transversal (EOCT)}, {\sc Split Vertex Deletion (SVD)} and {\sc Split Edge Deletion (SED)}.  {\sc Vertex Cover}, in particular, has been extensively studied in the parameterized complexity, and the current fastest algorithm runs in time $1.2738^kn^{\Oh(1)}$ and has a kernel with $2k$ vertices~\cite{DBLP:conf/mfcs/ChenKX06}. 
The parameterized complexity of  {\sc OCT} was a well known open problem for a long time. In 2003,  in a breakthrough paper,  Reed et al.~\cite{ReedSV04}  showed that {\sc OCT}  is {\FPT} by developing an algorithm for the problem running in time $\Oh(3^kmn)$.  In fact, this was the first time that the iterative compression technique was used. However, the algorithm for {\sc OCT} had seen no further improvements in the last $9$ years, though several reinterpretations of the algorithm have been published~\cite{Huffner09,LokshtanovSS09}. Only recently,  Lokshtanov et al.~\cite{LokshtanovNRRS14} obtained a faster algorithm for the problem running in time $2.3146^kn^{\Oh(1)}$ using a branching algorithm based on linear programming.   Guo et al.~\cite{GuoGHNW06} designed an algorithm for {\sc EOCT} running in time $2^kn^{\Oh(1)}$.
There is another theme of research in parameterized complexity, where the objective is to minimize the dependence of $n$ at the cost of a slow growing 
function of $k$.  A well known open problem, in the area, is whether {\sc OCT} admits a linear time parameterized algorithms. Only recently, the first linear time FPT algorithms for {\sc OCT} on general graphs were obtained, both of which run in time $\Oh(4^kk^{\Oh(1)}(m+n))$~\cite{RamanujanS14,IwataOY14}.  Kratsch and Wahlstr\"{o}m~\cite{KratschW14} obtained a randomized polynomial kernel for  {\sc OCT} and {\sc EOCT}. Ghosh et al.~\cite{DBLP:conf/swat/GhoshKKMPRR12} studied {\sc SVD} and {\sc SED} and designed algorithms with running time 
$2^kn^{\Oh(1)}$ and $2^{\Oh(\sqrt{k} \log k)}n^{\Oh(1)}$. They also gave the best known polynomial kernel for these problems. Later, Cygan and  Pilipczuk~\cite{DBLP:journals/ipl/CyganP13} designed an algorithm for {\sc SVD} running in time $1.2738^{k+o(k)}n^{\Oh(1)}$. Krithika and Narayanaswamy~\cite{DBLP:journals/jgaa/KrithikaN13} studied {\abpartization} problems on perfect graphs, and among several results they obtain $(r+1)^k n^{\Oh(1)}$ algorithm for {\sc Vertex $(r,0)$-Partization}  on perfect graphs.

\medskip

\noindent 
{\bf Our Results and Methods.}  We do not hope to get parameterized algorithms for either {\sc Vertex $(r,\ell)$-Partization} or {\sc Edge $(r,\ell)$-Partization} when either of $r$ or $\ell$ is at least $3$ as the recognition problem itself is \NP-complete. This leaves the case of $r,\ell \in  \{0,1,2\}$.  
We almost complete the parameterized complexity dichotomy for these problems by either obtaining new results or using the existing results. We refer to Figures~\ref{fig:vertexresults} and \ref{fig:edgeresults} for a summary of new and old results. 
%

\begin{figure}[t]
\centering
\begin{tabular}{|c|c|c|c|}
\hline
$r,\ell$                                                                                                & Problem Name                                                                                                                                                                      & FPT                                        & Kernel                                                                                            \\ \hline
$(1,0)$                                                                                              & {\sc Vertex Cover}                                                                                                                                                              & $1.2738^k$                                   & Poly                                                                                              \\ \hline
$(0,1)$                                                                                              & {\sc Vertex Cover} on $\overline{G}$                                                                                                                                            & $1.2738^k$                                 & Poly                                                                                              \\ \hline
$(1,1)$                                                                                              & {\sc SVD}                                                                                                                                                                       & $1.2738^{k+o(k)}$                          & Poly                                                                                              \\ \hline
$(2,0)$                                                                                              & {\sc OCT}                                                                                                                                                                       & $2.3146^k$                                & Randomized Poly                                                                                   \\ \hline
$(0,2)$                                                                                              & {\sc OCT on $\overline{G}$}                                                                                                                                                     & $2.3146^k$                                & Randomized Poly                                                                                   \\ \hline
{\color[HTML]{036400} \textbf{\begin{tabular}[c]{@{}c@{}}$(2,1)$, $(1,2)$, \\ $(2,2)$\end{tabular}}} & {\color[HTML]{036400} \textbf{\begin{tabular}[c]{@{}c@{}}\sc{Vertex $(2,1)$-partization}\\ \sc{Vertex $(1,2)$-partization}\\ \sc{Vertex $(2,2)$-partization}\end{tabular}}} & {\color[HTML]{036400} \textbf{$3.3146^k$}} & {\color[HTML]{036400} \textbf{\begin{tabular}[c]{@{}c@{}}Randomized \\ Turing Poly\end{tabular}}} \\ \hline
\end{tabular}

\caption{\label{fig:vertexresults}Summary of known and new results for the family of \abpartization{} problems. New results are highlighted in green (last row).}
\end{figure}

\begin{figure}[t]
\centering
\begin{tabular}{|c|c|c|c|}
\hline
$r,\ell$                                   & Problem Name                                                    & FPT                                          & Kernel                               \\ \hline
$(1,0)$                                 & \multicolumn{3}{c|}{\textit{~~~~~~~~~~~~~~~~~~~~~~~Recognizable in polynomial time.~~~~~~~~~~~~~~~~~~~~~~~}}                                                                                        \\ \hline
$(0,1)$                                 & \multicolumn{3}{c|}{\textit{Recognizable in polynomial time.}}                                                                                        \\ \hline
$(1,1)$                                 & {\sc SED}                                                     & $2^{\Oh(\sqrt{k} \log k)}$                             & Poly                                 \\ \hline
$(2,0)$                                 & {\sc EOCT}                                                    & $2^k$                                        & Randomized Poly                      \\ \hline
$(0,2)$                                 & \multicolumn{3}{c|}{\textit{Recognizable in polynomial time.}}                                                                                        \\ \hline
{\color[HTML]{036400} \textbf{$(2,1)$}} & {\color[HTML]{036400} \textbf{\sc{Edge $(2,1)$-partization}}} & {\color[HTML]{036400} \textbf{$2^{k+o(k)}$}} & {\color[HTML]{CE6301} \textbf{Open}} \\ \hline
{\color[HTML]{036400} \textbf{$(1,2)$}} & {\color[HTML]{036400} \textbf{\sc{Edge $(1,2)$-partization}}} & {\color[HTML]{036400} \textbf{FPT}}          & {\color[HTML]{CE6301} \textbf{Open}} \\ \hline
{\color[HTML]{036400} \textbf{$(2,2)$}} & {\color[HTML]{036400} \textbf{\sc{Edge $(2,2)$-partization}}} & \multicolumn{2}{c|}{{\color[HTML]{CE6301} \textbf{Open}}}                           \\ \hline
\end{tabular}
\caption{\label{fig:edgeresults}Summary of known and new results for the family of \abedgepartization{} problems. New results are highlighted in green.}
\end{figure}

%
%

For both {\sc Vertex $(r,\ell)$-Partization} and {\sc Edge $(r,\ell)$-Partization}, the  only new cases  for which we need to design new parameterized algorithms to complete the dichotomy is when $r,\ell \in  \{1,2\}$. Apart from the algorithmic results indicated in the Figures~\ref{fig:vertexresults} and \ref{fig:edgeresults}, we also obtain the following results. When $r,\ell \in \{1,2\}$,  we obtain an $\Oh(\sqrt{ \log{n}})$-approximation for these special cases. Finally, we obtain randomized  {\em Turing kernels} for \abpartization{} using this approximation algorithms. In particular, we give a polynomial time algorithm that produces polynomially many instances, $n^{\Oh(1)}$ of \abpartization{} of size $k^{\Oh(1)}$ such that with very high probability $(G,k)$ is a \YES\ instance of \abpartization{} if and only of one of the polynomially many instances of \abpartization{} of size $k^{\Oh(1)}$ is a \YES\ instance. 
 The question of existence of polynomial kernels for these special cases as well as for  {\sc Edge $(r,\ell)$-Partization} is open. Even the parameterized complexity of  {\sc Edge $(2,2)$-Partization} remains open.

%
\smallskip

\noindent
{\em Our methods.}
Most of the \FPT{} algorithms are based on the iterative compression technique and  use an algorithm for either {\sc OCT} or {\sc EOCT} as a subroutine.  One of the algorithms also uses methods developed in~\cite{MarxOR13}. 
To arrive at the approximation algorithm, we needed to take a detour. We start by looking at a slightly larger class of graphs called {\em $(r,\ell)$-split graphs}. A graph $G$ is an \absplitGraph{} if its vertex set can be partitioned into $V_1$ and $V_2$ such that the size of a largest clique in $G[V_1]$ is bounded by $r$ and the size of the largest independent set in $G[V_2]$ is bounded by $\ell$. Such a bipartition for the graph $G$ is called as \absplitPartition. The notion of \absplitGraph{s} was introduced in~\cite{Gyarfas98}. 
For any fixed $r$ and $\ell$, there is a finite forbidden set $\F_{r,\ell}$ for \absplitGraph{s}~\cite{Gyarfas98}. That is,  
a graph $G$ is a \absplitGraph{} if and only if $G$ does not contain any graph $H\in \F_{r,\ell}$ as an induced subgraph. The size of the largest forbidden graph is bounded by $f(r,\ell)$, $f$ being a function given in \cite{Gyarfas98}.
Since the class \abGraph{s} is a sub class of \absplitGraph{s}, each graph in $\F_{r,\ell}$ will not appear as an induced  subgraph in any \abGraph.  For our approximation algorithm we first make the given graph \absplitGraph\ by removing the induced subgraphs that are isomorphic to some graph in  $\F_{r,\ell}$. 
Once we have \absplitGraph, we generate a \absplitPartition\ $(V_1,(V_2)$ of $G$. Then we observe that for $r,\ell \in \{1,2\}$ the problem reduces to finding an approximate solution to  \OCT{} in $G[V_1]$ and $\overline{G}[V_2]$. Finally, we use the known $\Oh(\sqrt{\log{n}})$-approximation algorithm for \OCT{} \cite{AgarwalCMM05} to obtain a $\Oh(\sqrt{\log{n}})$-approximation algorithm for our problems. The Turing kernel for \abpartization, when $r,\ell \in \{1,2\}$, uses the approximation algorithm and depends on the randomized kernelization algorithm for \OCT{} \cite{KratschW14}.

\section{Preliminaries}

  We use standard notations from graph theory(\cite{diestel}) throughout this paper. 
The vertex set and edge set of a graph are  denoted as $V(G)$ and $E(G)$ respectively. The complement of the graph $G$, denoted by $\overline{G}$, is such that $\overline{G} = (V(G), E(C_{\vert V \vert})- E(G))$, where $C_n$ denotes a clique on $n$ vertices. The neighbourhood of a vertex $v$ is represented as $N_G(v)$, or, when the context of the graph is clear, simply as $N(v)$. An induced subgraph of $G$ on the vertex set $V'\subseteq V$ is written as $G[V']$. An induced subgraph of $G$ on the edge set $E' \subseteq E$ is written as $G[E']$. For a vertex subset $V' \subseteq V$, $G[V- V']$ is also denoted as $G - V'$. Similarly, for an edge set $E' \subseteq E$, $G - E'$ denotes the subgraph $G' = (V, E\setminus E')$.

The {\sc Ramsey number} for a given pair of positive integers $(a,b)$ is the minimum number such that any graph with the Ramsey number of vertices either has an induced independent set of size $a$ or an induced clique of size $b$. The Ramsey number for $(a,b)$ is denoted by $R(a,b)$.

We have already seen what \abGraph{s} are. Below, is a formal definition of the graph class as well as some related definitions.
\begin{definition}{\abGraph}
 A graph $G$ is an \abGraph{} if its vertex set can be partitioned into $r$ independent sets and $\ell$ 
cliques. We call such a partition of $V(G)$ an \abPartition. An \ICpartition, of an \abGraph{} $G$, is a partition $(V_1,V_2)$ of 
$V(G)$ such that $G[V_1]$ can be partitioned into $r$ independent sets and $G[V_2]$ can be partitioned into $\ell$ cliques.   
\end{definition}
For fixed $r,\ell \geq 0$, the class of \abGraph{s} is closed under induced subgraphs.

The following observation is useful in the understanding of the algorithms presented in the paper

\begin{observation}\label{ab_intersection}
 Let $P=(P_I,P_C)$ and $P'=(P'_I,P'_C)$ be two \ICpartitions{} of an \abGraph{} $G$. Then  
 $|P_I \cap P_C'| \leq r\ell$ and $|P_I' \cap P_C| \leq r\ell$.
\end{observation}

\begin{proof}
Consider an independent set $I \in P_I$ and a clique $C \in P_C'$. At most $1$ vertex of $C$ can also be contained in $I$. There are at most $r$ independent sets in $P_I$ 
and so $P_I$ can contain at most $r$ vertices from $C$. There are at most $\ell$ cliques in $P_C'$ each of which can have an intersection of at most $r$ vertices with $P_I$.
Hence, $|P_I \cap P_C'| \leq r\ell$. Similarly, we can prove that  $|P_I' \cap P_C| \leq r\ell$. \qed
\end{proof}

\section{Vertex Deletion for \abGraph{s}}
\label{sec:verdelabg}
In this section we first show that \twotwopartization{} is in \FPT, using  iterative compression. 
Then we explain how to reduce \twoonepartization{} and \onetwopartization{} to \twotwopartization. 
Our algorithm for \twotwopartization{} combines the iterative compression technique with a polynomial bound on the number of \ICpartitions{} of a \twotwoGraph. The following
Lemma tells about an algorithm to recognize whether a graph is a \twotwoGraph{} and also about an algorithm to compute all such \ICpartitions. These results were shown in several papers \cite{Brandstadt98thecomplexity,Feder03listpartitions}.

\begin{lemma}
\label{lemma:boundpartition}
Given a graph $G$ on $n$ vertices and $m$ edges we can recognize whether $G$ is a \twotwoGraph{} in $\Oh((n+m)^2)$ time. Also, a \twotwoGraph{} can have at most $n^8$ \ICpartitions{} and all the \ICpartitions{} can be enumerated in $\Oh(n^8)$ time.
\end{lemma}

For a graph $G$, we say $S\subseteq V(G)$ is a \twotwoDeletion, if $G- S$ is a \twotwoGraph. 
Now we describe the iterative compression technique and its application to the \twotwopartization{} problem.

\smallskip
\noindent
\textbf{Iterative Compression for \twotwopartization.} 
Let $(G,k)$ be an input instance of \twotwopartization{} and let $V(G) =\{v_1,\ldots,v_n\}$. We define, for every $1\leq i \leq \vert V(G) \vert$, the vertex set
$V_i=\{v_1,\ldots,v_i\}$. 
Denote $G[V_i]$ as $G_i$. We iterate through the instances $(G_i,k)$ starting from $i = k+5$. Given the $i^{th}$ instances and a known \twotwoDeletion{} $S'_i$ of size at most
$k+1$, our objective is to obtain a \twotwoDeletion{} $S_i$ of size at most $k$. The formal definition of this compression problem is as follows.
 
\defparproblemoutput{{\sc \twotwopartization{} Compression}}{A graph $G$ and a $k+1$ sized vertex subset $S'\subseteq V(G)$ such that $G- S'$
is a \twotwoGraph}{$k$}{A vertex subset $S \subseteq V(G)$  of size at most $k$ such that $G- S$
is a \twotwoGraph?}

We reduce the \twotwopartization{} problem to $n-k-4$ instances of the {\sc \twotwopartization{} Compression} problem in the following manner. 
When $i = k+5$, the set $V_{k+1}$ is a \twotwoDeletion{} of size at most $k+1$ for $G_{k+5}$. Let $I_i = (G_i,S'_i,k)$ be the $i^{th}$ instance of {\sc \twotwopartization{}
Compression}. If $S_{i-1}$ is a $k$-sized solution for $I_i$, then $S_{i-1} \cup \{v_i\}$ is a $(k+1)$-sized \twotwoDeletion{} for $G_i$. Hence, we start the iteration with the
instance $I_{k+5} = (G_{k+5}, V_{k+1},k)$ and try to obtain a \twotwoDeletion{} of size at most $k$. If such a solution $S_{k+5}$ exists, we set $S'_{k+5} = S_{k+5} \cup
\{v_{k+6}\}$ and ask of a $k$-sized solution for the instance $I_{k+6}$, and so on. If, during any iteration, the corresponding instance does not have a \twotwoDeletion{} of size at most 
$k$, it implies that the original instance $(G,k)$ is a \NO{} instance for \twotwopartization. If the input instance $(G,k)$ is a \YES{} instance, then $S_n$ is a $k$-sized
\twotwoDeletion{} for $G$, where $n=\vert V(G)\vert$. Since there are at most $n$ iterations, the total time taken by the algorithm to solve \twotwopartization{} is at most $n$
times the time taken to solve {\sc \twotwopartization{} Compression}. The above explained template for doing iterative compression will be used for approximation algorithms as well as for parameterized algorithms for edge versions of these problems. 

Next we show that {\sc \twotwopartization{} Compression} is in \FPT. the arguments above imply that \twotwopartization{} is also in \FPT. 
\begin{lemma}
\label{lem:22comp}
{\sc \twotwopartization{} Compression} can be solved deterministically in time $3.3146^{k}\vert V(G)\vert ^{\Oh(1)}$. 
\end{lemma}
\begin{proof}
We design an algorithm for {\sc \twotwopartization{} Compression}. Let 
$(G,S')$ be the instance of the problem and let $(P_I',P_C')$ be an \ICpartition{} of $G - S'$. Let $S$ be an {\em hypothetical solution} of size $k$ for the problem, which the 
algorithm suppose to compute. Let $(P_I,P_C)$ be an \ICpartition{} of $G- S$. The algorithm first guesses a partition $(Y,N)$ of $S'$ 
such that $Y=S'\cap S$ and $N=S'- S$. 
After this guess, the objective is to compute a set $Z$ of size at most $k'=k-\vert Y\vert$ such that $G- (Z\cup Y)$ is 
a \twotwoGraph. Also note that since $N$ is not part of the solution $S$, $G[N]$  is a \twotwoGraph. 
Consider the two \ICpartitions{} $(P_I- (S\cup S'),P_C- (S\cup S'))$ and $(P_I'- (S\cup S'),P_C'- (S\cup S'))$  
of the \twotwoGraph{} $G- (S\cup S')$. By Observation~\ref{ab_intersection} we know that the cardinality of each of the set $(P_I\cap P_C')- (S\cup S')$ and 
$(P_C\cap P_I')- (S\cup S')$ are bounded by $4$. So now the algorithm guesses the set $V_I=(P_I\cap P_C')- (S\cup S')$ 
and $V_C=(P_C\cap P_I')- (S\cup S')$, each of them having size at most $4$. After the guess of $V_I$ and $V_C$, any vertex in $P_C'- V_I$ either belongs to $P_C$ or belongs to the hypothetical solution $S$. Similarly 
any vertex in $P_I'- V_C$ either belongs to $P_I$ or belongs to the hypothetical solution $S$.
By Lemma~\ref{lemma:boundpartition} we know that the number of \ICpartitions{} of $G[N]$ is at most $\Oh(k^8)$ 
and  
these partitions can be enumerate in time $\Oh(k^8)$ . The algorithm now guesses an \ICpartition{} $(N_I,N_C)$ of $G[N]$ such that 
$N_I\subseteq P_I$ and $N_C\subseteq P_C$. 
Now consider the partition $(A,B)=((P_I'\cup N_I \cup V_I)- V_C, (P_C'\cup N_C \cup V_C)- V_I)$. 
Any vertex $v\in A$ either belongs to $P_I$ or belongs to the hypothetical solution $S$ and any 
any vertex $v\in B$ either belongs to $P_C$ or belongs to the solution $S$. 
So the objective is to find two sets $U\subseteq A$ and $W\subseteq B$ such that $G[A- U]$ 
is a bipartite graph, $G[B- W]$ is the complement of a bipartite graph and $\vert U\vert +\vert W\vert \leq k'$.
As a consequence, the algorithm guesses the sizes $k_1$ of $U$ and $k_2$ of $W$. 
Then the problem reduced to finding an odd cycle transversal(OCT) of size $k_1$ for $G[A]$ 
and an OCT of size $k_2$ for the complement of the graph $G[B]$.  
Hence, our algorithm runs the current best 
algorithm for \OCT{}, 
presented in~\cite{LokshtanovNRRS14} 
for finding an OCT $U$ of size $k_1$ in $G[A]$ and for finding an OCT $W$ of size $k_2$ in the complement of $G[B]$. 
The running times of the \OCT{} algorithm on $G[A]$ and on the complement of $G[B]$ are 
$2.3146^{k_1}\vert V(G)\vert^{\Oh(1)} $ and $2.3146^{k_2} \vert V(G)\vert^{\Oh(1)}$ respectively. 
Finally, our algorithm 
outputs $Y\cup U\cup W$.

\smallskip

\noindent
\textbf{Running Time.} 
Let $n=\vert V(G)\vert$. 
The algorithm guesses the set $Y=S\cap S'$. First we fix a set $Y$ of size $k-i$ and compute the running time 
of this particular guess. The algorithm guess $V_C$ and $V_I$ each of size at most $4$. The number of such 
guesses is bounded by $\Oh(n^8)$. Our algorithm also guess a partition $(N_I,N_C)$ of $S' - Y$. By Lemma~\ref{lemma:boundpartition} 
number of such guesses are bounded by $k^8$. At the end, the algorithm guesses $k_1$ and $k_2$ 
such that $k_1+k_2 = k-\vert Y\vert =i$. Then our algorithm executes algorithm for \OCT{} 
for two instances, running in time $2.3146^{k_1}n^{\Oh(1)}$ and $2.3146^{k_2}n^{\Oh(1)}$ 
Thus the running time for a particular guess $Y$ is bounded by $2.3146^{i}n^{\Oh(1)}$. 
The number of guesses for $Y$ of size $i$ is exactly ${k+1 \choose i}$. 
Since $\sum_{i=0}^{k+1} {k+1 \choose i} 2.3146^{i}n^{\Oh(1)} = 3.3146^{k}n^{\Oh(1)}$, the total running time is bounded by $3.3146^{k}n^{\Oh(1)}$. \qed
\end{proof}

Lemma~\ref{lem:22comp} and the discussions preceding it imply the following theorem. 
\begin{theorem}
\label{thm:22verdel} 
\twotwopartization{} can be solved in time $3.3146^{k}\vert V(G)\vert ^{\Oh(1)}$. 
\end{theorem}

\noindent
\textbf{\twoonepartization:}
There is a simple reduction from the \twoonepartization{} problem to the \twotwopartization{} problem. Suppose we are given a graph $G$, where $\vert V(G) \vert = n$. We construct a graph $G' = G \uplus \hat{C}$, where $\hat{C}$ is a clique on $n+3$ new vertices.  That is, $G'$ is the disjoint union of $G$ and $\hat{C}$. The next lemma relates the graphs $G$ and $G'$.

\begin{lemma}\label{twoone_twotwo}
 For any integer $t\leq n$, $(G,t)$ is a \YES{} instance of \twoonepartization{} if and only if $(G',t)$ is a \YES{} instance of \twotwopartization. 
\end{lemma} 

\begin{proof}
 Suppose $(G,t)$ is a \YES{} instance of \twoonepartization{}. Then there is a subset $S \subseteq V(G)$, of size at most $t$, the deletion of which results in a \twooneGraph{} $G^*$. Let $G^*$ have a \twoonePartition{} $I_1\cup I_2 \cup C_1$. Then, $I_1 \cup I_2 \cup C_1 \cup \hat{C}$ is a \twotwoPartition{} for $G' - S$. Hence, $(G',t)$ is a \YES{} instance of \twotwopartization.

Conversely, suppose $(G',t)$ is a \YES{} instance of \twotwopartization. Let $S \subseteq V(G')$ be a \twotwoDeletion{} of size at most $t$. The deletion of $S$ from $G'$ results in a
\twotwoGraph{} $\tilde{G}$. Let $\tilde{G}$ have a \twotwoPartition{} $I_1\cup I_2 \cup C_1 \cup C_2$. Since $t \leq n$, and since any independent set $I$ of $G'$ can have at most $1$ vertex
from $\hat{C}$, $\vert \hat{C}- (S\cup I_1 \cup I_2)\vert \geq n-t+3$. As $\hat{C}$ is disjoint from $G$, it is only possible that either $C_1 \subseteq \hat{C}$ and $C_2
\cap \hat{C} = \emptyset$ or $C_2 \subseteq \hat{C}$ and $C_1 \cap \hat{C} = \emptyset$. Without loss of generality, suppose $C_1 \subseteq \hat{C}$ and $C_2 \cap \hat{C} = \emptyset$. Then
$S'= S - \hat{C}$ is of size at most $t$ and $G- S'$ has a \twoonePartition{} $(I_1- \hat{C}) \cup (I_2 - \hat{C})\cup C_2$. Thus, $(G,t)$ is a \YES{}
instance of \twoonepartization. \qed  
\end{proof}

Now if we are given an instance $(G,k)$ of \twoonepartization, Lemma~\ref{twoone_twotwo} tells us that it is enough to solve \twotwopartization{} on $(G',k)$. Notice that solving the \onetwopartization{} problem on an input instance $(G,k)$ is equivalent to finding a \onetwopartization{} on $(\overline{G},k)$, where $\overline{G}$ is the complement graph of $G$. Thus, we get the following as a corollary of Theorem~\ref{thm:22verdel}.

\begin{corollary}
 \onetwopartization{} and \twoonepartization{} have \FPT{} algorithms that run in $3.3146^{k}n^{\Oh(1)}$ time.
\end{corollary}
 
%
%
%

\section{ Approximation algorithms for Vertex Deletion to \abGraph{s}}
In this section we give a polynomial time approximation algorithm for \twotwopartization. 
That is, we design an algorithm for \twotwopartization, which takes an instance 
$(G,k)$, runs in polynomial time and outputs either a solution of size $\Oh(k^{3/2})$ or concludes that 
$(G,k)$ is a \NO{} instance. 
Since the reduction from  \twoonepartization{} to \twotwopartization{}, given in Lemma~\ref{twoone_twotwo}, is  an 
approximation preserving reduction, we can get a similar approximate algorithm for \twoonepartization. Similarly, since \onetwopartization{} on a graph is equivalent to \twoonepartization{} in the complement graph, we can get an approximation algorithm for \onetwopartization. 
The approximation algorithm we discuss in this section, is useful for obtaining Turing kernels for
\abpartization, when $1\leq r,\ell \leq 2$. Finally, we design a factor 
$\Oh(\sqrt{\log n})$ approximation algorithms for these problems.

First we define superclass of \abGraph{s}, called {\em \absplitGraph{s}} and then design a polynomial time recognition algorithm for 
\absplitGraph{s}, which is used for approximation algorithm for \twotwopartization.   
The notion of \absplitGraph{s} was introduced in~\cite{Gyarfas98}. 
\begin{definition}[\absplitGraph]
 A graph $G$ is an \absplitGraph{} if its vertex set can be partitioned into $V_1$ and $V_2$ such that the size of a largest clique in $G[V_1]$ is bounded by $r$
 and the size of a largest independent set in $G[V_2]$ is bounded by $\ell$. We call such a bipartition for the graph $G$ an 
\absplitPartition.
\end{definition}
Now we give a polynomial time algorithm which takes a graph $G$ as input and outputs an \absplitPartition{} if $G$ is an \absplitGraph. 
We design such an algorithm using iterative compression. Essentially we show that 
the following problem, {\sc \absplitPartition{} Compression}, can be solved in 
polynomial time. 

\defproblem{{\sc \absplitPartition{} Compression}}{A graph $G$ with $V(G)=V\cup \{v\}$  
and an \absplitPartition{} $(A,B)$ of $G[V]$ }{An \absplitPartition{} of $G$, if $G$ is an \absplitGraph{}, and \NO{} otherwise} 

Like in the case of the \FPT{} algorithm for \twotwopartization{} given in Section~\ref{sec:verdelabg}, we can show that 
by running the algorithm for {\sc \absplitPartition{} Compression} at most $n-2$ times we can get an algorithm 
which outputs an \absplitPartition{} of a given \absplitGraph. 
Our algorithm for {\sc \absplitPartition{} Compression} uses the following simple lemma. 

\begin{lemma}
 \label{lemma:split_partition_intersection}
 Let $G$ be an \absplitGraph. Let $(A,B)$ and $(A',B')$ are two \absplitPartition{s} of $G$. 
 Then $\vert A\cap B'\vert \leq R(\ell+1,r+1)-1$ and $\vert A'\cap B \vert \leq R(\ell+1,r+1)-1$, 
 where $R(r+1,\ell+1)$, is the Ramsey number. 
\end{lemma}
\begin{proof}
 Suppose $\vert A\cap B'\vert\geq R(\ell+1,r+1)$. By Ramsey's theorem, we know that $G[A\cap B']$ either contain an 
 independent set of size $\ell+1$ or a clique of size $r+1$. If $G[A\cap B']$ contains an independent set of size $\ell+1$, then 
 it contradicts our assumption that $(A',B')$ is an \absplitPartition{} of $G$. 
 If $G[A\cap B']$ contains a clique of size $r+1$, then 
 it contradicts our assumption that $(A,B)$ is an \absplitPartition{} of $G$. 
 This implies that $\vert A\cap B'\vert\leq R(\ell+1,r+1)-1$. 
 By similar arguments we can show that $\vert A'\cap B\vert\leq R(\ell+1,r+1)-1$.  \qed 
\end{proof}
Using Lemma~\ref{lemma:split_partition_intersection}, we show that {\sc \absplitPartition{} Compression} 
can be solved in polynomial time for any fixed constants $r$ and $\ell$. 
\begin{lemma}
 \label{lemma:polytime_split_compression}
 For any fixed constants $r$ and $\ell$, {\sc \absplitPartition{} Compression} 
 can be solved in polynomial time.
\end{lemma}
\begin{proof}
 Let $(G,(A,B))$ be the given instance of {\sc \absplitPartition{} Compression}, where 
 $(A,B)$ is a \absplitPartition{} of $G[V]$. Let $n=\vert V(G) \vert$. Let $(A',B')$ be a hypothetical solution for the problem. 
 Since $G[V]$ is a subgraph of $G$, $(A'\setminus \{v\},B'\setminus\{v\})$ is an \absplitPartition{} of 
 $G[V]$. Thus, by Lemma~\ref{lemma:split_partition_intersection}, we know that 
 $\vert A\cap B'\vert \leq R(\ell+1,r+1)-1$ and $\vert A'\cap B \vert \leq R(\ell+1,r+1)-1$. 
 So our algorithm guesses the sets $U=A\cap B'$ and $W=A'\cap B$ each of size at most $R(\ell+1,r+1)$. 
 The total number of possible choices for $U$ and $W$ is clearly bounded by $n^{2R(\ell+1,r+1)}$. 
 For the correct guess $U$ and $W$, $A'\setminus \{v\}= (A\cup W)\setminus U $ and $B'\setminus \{v\}=(B\cup U)\setminus W$.
 Let $X=(A\cup W)\setminus U $ and $Y=(B\cup U)\setminus W$. 
 So now it is enough to check whether one of the $(X\cup \{v\},Y)$ or $(X,Y\cup \{v\})$ is 
 a valid \absplitPartition{} of the graph $G$ and output the result. This can be tested in time $n^{r+\ell}$ time. 
 Since there are $n^{2R(\ell+1,r+1)}$ choices for the guess $U$ and $W$, the total running time is 
 bounded by $\Oh(n^{2R(\ell+1,r+1)+r+\ell})$. This completes the proof of the lemma. \qed
\end{proof}
By applying Lemma~\ref{lemma:polytime_split_compression}, at most $n-2$ times, we can get the following lemma. 
\begin{lemma}
 \label{lemma:polytime_split_reco}
 For any fixed constants $r$ and $\ell$, there is an algorithm which takes a graph $G$ as input, 
 runs in polynomial time,  and decides whether $G$ is an \absplitGraph. Furthermore, if $G$ is an \absplitGraph{} then the algorithm
outputs an \absplitPartition{} $(V_1,V_2)$ of $G$ 
\end{lemma}

We know that any \abGraph{} is also an \absplitGraph. 
The following lemma gives a relation between an \absplitPartition{} and \ICpartition{} of a \abGraph. 
\begin{lemma}
 \label{lemma:rel_split_ICpartition}
 Let $G$ be an \abGraph. Let $(A,B)$ be an \ICpartition{} of $G$ and $(A', B')$ be 
 an \absplitPartition{} of $G$. Then $|A \cap B'| \leq r\ell$ and $|A' \cap B| \leq r\ell$
\end{lemma}
\begin{proof}
 Suppose $|A \cap B'| \geq r\ell+1$. Since $(A,B)$ is an \ICpartition{} of an \abGraph{} $G$, 
 we know that $A$ can be partitioned into $r$ independent sets. Also, since $|A \cap B'| \geq r\ell+1$, 
 by pigeon hole principle, there is an independent set $I$ in $A$ such that $\vert I\cap B'\vert \geq \ell+1$. 
 This implies that the size of the largest independent set in $B'$ is at least $\ell+1$, contradicting 
 our assumption that $(A', B')$ is an \absplitPartition{} of $G$. Hence we have shown that $|A \cap B'| \leq r\ell$. 
 By similar arguments we can show that $|A' \cap B| \leq r\ell$. \qed
\end{proof}

Before giving an approximation algorithm for \abpartization{}, we need to mention about 
a polynomial time approximation algorithm for {\sc Odd Cycle Transversal} and finite forbidden characterization 
of \abGraph{s}. 
Using the FPT algorithm for OCT~\cite{KratschW14}, and a $\Oh(\sqrt{\log n})$-approximation algorithm for 
OCT~\cite{AgarwalCMM05}, one can prove the following proposition. 
\begin{proposition}[\cite{KratschW14}]\label{prop:oct_algo}
There is a polynomial time algorithm which takes a graph $G$ and an integer $k$ as input and 
outputs either an OCT of $G$ of size at most $\Oh(k^{3/2})$ or concludes that 
there is no OCT of size $k$ for $G$. 
\end{proposition}
For any fixed $r$ and $\ell$, there is a finite forbidden set $\F_{r,\ell}$ for \absplitGraph{s}~\cite{Gyarfas98}. That is,  
a graph $G$ is an \absplitGraph{} if and only if $G$ does not contain any graph $H\in \F_{r,\ell}$ as an induced subgraph. The size of the largest forbidden graph is bounded by $f(r,\ell)$, $f$ being a function given in \cite{Gyarfas98}. Since $f(2,2)$ is a constant, it is possible to compute the forbidden set $\F_{r,\ell}$ in polynomial time: The forbidden graphs are of size at most $f(2,2)$. 
Since the class \abGraph{s} is a sub class of \absplitGraph{s}, each graph in $\F_{r,\ell}$ will not appear as an induced  subgraph in any \abGraph.  
Now we are ready to design a polynomial time approximation algorithm for \twotwopartization. 
\begin{theorem}
 \label{lemma:fpt_approx_22}
 There is an algorithm which takes a graph $G$ and an integer $k$ as input, runs in polynomial time and 
 outputs either a set $S$ of size $\Oh(k^{3/2})$ such that $G-S$ is a \twotwoGraph{} or concludes that 
 $(G,k)$ is a \NO{} instance of \twotwopartization. 
\end{theorem}
\begin{proof}
 The algorithm first finds a maximal set ${\cal T}$ of vertex disjoint subgraphs of $G$ such that each subgraph in ${\cal T}$ 
 is isomorphic to a graph in $\F_{2,2}$. If $\vert {\cal T}\vert>k$, then clearly $(G,k)$ is a \NO{} instance of 
 \twoonepartization. So the algorithm will output \NO{} if $\vert {\cal T}\vert>k$. 
 Now consider the graph $G'=G-V({\cal T})$. Here, $V({\cal T})$ denotes the set of vertices appearing in graphs in $\cal T$. 
 Since $\cal T$ is a maximal set of vertex disjoint subgraphs in $G$ which are  
 isomorphic to a graphs in $\F_{2,2}$ we have that $G'$ is a \twotwosplitGraph. 
 
 Now our algorithm will find a set $S\subseteq V(G')$ of size bounded by $\Oh(k^{3/2})$ such that 
 $G'-S$ is a \twotwoGraph. 
 Since $G'$ is a subgraph of $G$, if $(G,k)$ is a \YES{} instance of \twoonepartization, then 
 $(G',k)$ is also a \YES{} instance. Let $S^*$ be an hypothetical solution of the instance $(G',k)$ of \twoonepartization{} and let $(A,B)$ be an \ICpartition{} of 
 $G'-S^*$. Now our algorithm applies Lemma~\ref{lemma:polytime_split_reco} on graph $G'$ and computes 
 a \twotwosplitPartition{} $(A',B')$ of $G'$ in polynomial time. 
 By Lemma~\ref{lemma:rel_split_ICpartition}, we know that $\vert A \cap B'\vert \leq 4$ and 
 $\vert A' \cap B\vert \leq 4$. So the algorithm will guess the set $U= A \cap B'$ and 
 $W=A' \cap B$. The number of possible guesses for $U$ and $W$ is bounded by $n^8$. 
 For the correct guess $U$ and $W$, we know that $A=(A'\cup U)\setminus (W\cup S^*)$ and $B=(B'\cup W)\setminus (U\cup S^*)$.
 Now consider the partition $(V_1,V_2)$ of $V(G')$, where $V_1=(A'\cup U)\setminus W$ and $V_2=(B'\cup W)\setminus U$. 
 So for the correct guess $U$ and $W$, we know that each vertex in $V_1$ either belongs to $A$ or belongs to $S^*$ 
 and each vertex in $V_2$ either belongs to $B$ or belongs to $S^*$. Now to compute a solution for $(G',k)$, it is enough to find an OCT 
 $S_1$ in $G[V_1]$ and an OCT $S_2$ in the complement graph of $G'[V_2]$ such that $\vert S_1\vert + \vert S_2\vert =k$. 
  Our algorithm applies Proposition~\ref{prop:oct_algo} on $G'[V_1]$ and on the complement graph of $G'[V_2]$. 
 If these algorithms output an OCT $S_1$ and an OCT $S_2$ for graphs $G'[V_1]$ and $\overline{G'}[V_2]$, 
 then $S_1\cup S_2$ is of size bounded by $\Oh(k^{3/2})$ and $G'-(S_1\cup S_2)$ is a \twotwoGraph. 
 Since $G'=G-V({\cal{T}})$ and $G'-(S_1\cup S_2)$ is a \twotwoGraph, we have that $G-(S_1\cup S_2\cup V({\cal T}))$ 
 is a \twotwoGraph. So our algorithm will output $S_1\cup S_2\cup V({\cal{T}})$ as the required output. 
 Since $|V({\cal{T}})|\leq k\cdot f(2,2)$, we have that $\vert S_1\cup S_2\cup V({\cal{T}})\vert = \Oh(k^{3/2})$. 
 If the algorithm mentioned in Proposition~\ref{prop:oct_algo} returns \NO{}  
 for all possible guesses of $U$ and $W$, then our algorithm outputs \NO. 
 It is easy to see that the number of steps in our algorithm is bounded by a polynomial in $\vert V(G)\vert$. \qed
\end{proof}

Using the arguments of Theorem~\ref{lemma:fpt_approx_22}, we can also design an approximation algorithm for finding a minimum \twotwoDeletion{} of a graph $G$. Let $S$ be an optimum \twotwoDeletion{} and  $(A,B)$ be the corresponding \ICpartition{} of 
 $G'=G-S$. Let ${\cal T}$ be a maximal set of vertex disjoint subgraphs of $G$, that are each isomorphic to a graph in $\F_{2,2}$. The number of subgraphs in ${\cal T}$ is at most $\vert S\vert$ and the number of vertices involved in these forbidden subgraphs is at most $f(2,2)\vert S\vert$. The remaining graph $G'$ is a \twotwosplitGraph{} and using 
 Lemma~\ref{lemma:polytime_split_reco}, we can find  a \twotwosplitPartition\ $(A',B')$ of $G'$. Let $(\hat{A},\hat{B})$ be the restriction of $(A,B)$ to $G'$. As argued above, at most $4$ vertices from $A'$ could be part of $\hat{B}$. Let this set of $4$ vertices be called $U$. The rest either belong to $\hat{A}$ or $S$. $U \cup (S\cap A')$ is an OCT for $A'$, of size at most $2\vert S\cap A'\vert$. The algorithm of~\cite{AgarwalCMM05} returns an $\Oh(\sqrt{\log{n}})$-approximate \OCT{} solution $S_1$ for $G[A']$, which has to be of size at most $2\vert S\cap A'\vert\cdot \Oh(\sqrt{\log{n}})$. 
There is a similar property on the vertices of $B'$. Applying the algorithm of~\cite{AgarwalCMM05}, on $\overline{G'[B']}$, returns an $\Oh(\sqrt{\log{n}})$-approximate \OCT{} solution $S_2$, which has to be of size at most $2\vert S\cap B'\vert\cdot \Oh(\sqrt{\log{n}})$.
Thus $V({\cal{T}}) \cup S_1\cup S_2$ is a \twotwoDeletion{} of $G$, with size at most $(f(2,2)+ \Oh(\sqrt{\log{n}})\vert S \vert$. This together with Lemma~\ref{twoone_twotwo} and discussion after that lead to the following theorem.

\begin{theorem}
\twoonepartization, \onetwopartization, and 
\twotwopartization\ admit polynomial time approximation algorithms with factor $\Oh(\sqrt{\log{n}})$. 
\end{theorem}

\section{Turing Kernels for Vertex Deletion to \abGraph{s}}
 In this section, we give a randomized Turing kernel for \twotwopartization\  (See introduction for the definition). 
 The equivalence in Lemma~\ref{twoone_twotwo} ensures that there is a randomized Turing kernel for \twoonepartization. Since, \onetwopartization{} on an instance $(G,k)$ is equivalent
to \twoonepartization{} on $(\overline{G},k)$, a randomized Turing kernel for \onetwopartization{} follows.
 
We have seen in Section~\ref{sec:verdelabg} that eventually the algorithm for \twotwopartization{} runs two instances of {\sc OCT}.  
In this section we explain that we can use the kernelization of {\sc OCT} to get a Turing kernel for \twotwopartization. 
A randomized polynomial kernel for {\sc OCT} was shown by Kratsch and Wahlstr\"{o}m~\cite{DBLP:conf/focs/KratschW12},  using the concept of representative family. They showed that 
it is possible to find $k^{\Oh(1)}$  ``relevant'' vertices from the input graph which contains the 
optimum solution. This leads to a randomized kernel for  {\sc OCT}. In fact, the following lemma follows from the work of Kratsch and Wahlstr\"{o}m. We sketch a proof in the appendix. 

\begin{lemma}
 \label{lemma:oct_strong}
 Let $G$ be a graph and $X$ be an OCT of $G$. There is a randomized polynomial time algorithm which computes a set $Z\subseteq V(G)- X$ 
 of size $\Oh(\vert X \vert^3 )$ such that 
 for any $Y\subseteq X$, a minimum sized OCT not containing $X$, of $G-Y$, is fully contained in $Z$. 
\end{lemma}

Now we are ready to explain our Turing kernel for \twotwopartization{} using Lemma~\ref{lemma:oct_strong}. 
Given an instance $(G,k)$ of \twotwopartization{}, first we construct  $\vert V(G)\vert^{\Oh(1)}$ many instances of a problem which is in \NP{} and each of them have size 
bounded by polynomial in $k$. Then, by using the Cook-Levin theorem~\cite{Cook:1971:CTP:800157.805047}, we can reduce each of these intances 
to instances of \twotwopartization{} and thus arrive at a Turing kernelization for \twotwopartization. 
We first run the polynomial time approximation algorithm described in Theorem~\ref{lemma:fpt_approx_22}. 
If the approximation algorithm outputs \NO{}, then the algorithm will output a trivial \NO{} instance 
of the problem. Otherwise let $X$ be the solution returned by the approximation algorithm on input $(G,k)$. We know that 
the cardinality of $X$ is bounded by $\Oh(k^{3/2})$. Now we fix an \ICpartition{} $(P_I,P_C)$ of $G-X$. 
Let $S$ be a {\em hypothetical} solution of size at most $k$ and $(Q_I,Q_C)$ be an \ICpartition{} of $G-S$. 
It follows from Observation~\ref{ab_intersection} that 
$\vert P_I\cap Q_C\vert\leq 4$ and $\vert Q_I\cap P_C\vert\leq 4$. 
This observation leads to the following lemma.  
\begin{lemma}\label{22_equivalence1}
 $(G,k)$ is a \YES{} instance of \twotwopartization{} if and only if there exist $V_C\subseteq P_I$ and $V_I\subseteq P_C$, each of cardinality at
most $4$ 
 such that $X'=X\cup V_C\cup V_I $ can be partitioned into $X'_I,X'_D,X'_C$, with the following properties:
\begin{enumerate}

\item There is a set $Z_I \subseteq P_I\setminus V_C$ such that $Z_I \cup X'_D \cup X'_C$ is an OCT for $G[P_I\cup X']$. 
In other words, $Z_I$ is an OCT for $G[P_I\cup X_I']$.
 \item 
There is a set $Z_C \subseteq P_C \setminus V_I$ such that $Z_C \cup X'_D \cup X'_I$ is an OCT for $\overline{G}[P_C
\cup X']$. In other words, $Z_C$ is an OCT for $\overline{G}[P_C \cup X'_C]$. 
\item $\vert Z_I\cup Z_C\cup X'_D\vert\leq k$.  
\end{enumerate}
 \end{lemma}

 \begin{proof}
  Suppose $(G,k)$ is a \YES{} instance of \twotwopartization. Then there is a $k$-sized solution $Z$ 
  such that $G-Z$ is a \twotwoGraph. Let 
  $(Q_I,Q_C)$ be an \ICpartition{} of $G-Z$. Let $V_C = P_I \cap Q_C$ and $V_I = P_C \cap Q_I$. 
 It follows from Observation~\ref{ab_intersection} 
 that that $\vert V_I\vert\leq 4$ and $\vert V_C\vert\leq 4$.  Notice that any vertex in $P_I
\setminus V_C$ either belongs to $Q_I$ or to $Z$. Similarly, any vertex in $P_C \setminus V_I$ either belongs to 
$Q_C$ or to $Z$. Let $X'=X\cup V_I\cup V_C$. Now we define $X_I'=X'\cap Q_I$, $X_C'=X'\cap Q_C$ and $X_D'=X'\cap Z$. 
Let $Z_I=Z\cap P_I$ and $Z_C=Z\cap P_C$. Note that $Z_I\cap V_C=\emptyset$ and $Z_C\cap V_I=\emptyset$.
From the definition of $X'$, $V_I$ and $V_C$, it is clear that $V_I\subseteq X_I'$ and $V_C\subseteq X_C'$. 
Since $V_C\subseteq X_I'$ and $V_C\subseteq X_C'$, we have that $(P_I\cup X')\setminus (Z_I\cup X_D'\cup X_C')=Q_I$. Also since, 
$G[Q_I]$ is a bipartite graph we have that $(Z_I\cup X_D'\cup X_C')$ is an OCT of $G[P_I\cup X']$. 
By similar arguments we can show that $(Z_C\cup X_D'\cup X_I')$ is an OCT of $\overline{G}[P_C\cup X']$. 
Since $Z_I\cup Z_C \cup X_D'=Z$ and $\vert Z\vert =k$, the set $Z_I\cup Z_C \cup X_D'$ satisfies condition 
3 in the lemma. This completes the proof of the forward direction. 

  
%
%
    
  Conversely, suppose there is a $V_C\subseteq P_I$ and $V_I \subseteq P_C$, each of size at most $4$ such that the $X' = X \cup V_I \cup V_C$ has a $3$-partition $(X'_I \cup X'_D
\cup X'_C)$ with the properties mentioned in the lemma. That is, there is an OCT $Z_I$ for the graph $G[P_I\cup X_I']$ and an OCT $Z_C$ 
for the graph $\overline{G}[P_C\cup X'_C]$ such that $\vert Z_I\cup Z_C \cup X_D'\vert \leq k$. Then we claim that $Z=Z_I\cup Z_C \cup X_D'$ is 
a \twotwoDeletion{} of $G$. Consider the sets $Q_I=(P_I\cup X_I')\setminus Z_I$ and $Q_C=(P_C\cup X_C')\setminus Z_C$. By our assumption 
$G[Q_I]$ and $\overline{G}[Q_C]$ are bipartite graphs. Also note that $Q_I\cup Q_C\cup Z=V(G)$. Hence $Z$ is a \twotwoDeletion{} of $G$ 
and $(Q_I,Q_C)$ is an \ICpartition{} of $G-Z$. \qed
%
%
 \end{proof} 
The Lemma~\ref{22_equivalence1} allows us to reduce an instance of \twotwopartization{} to polynomially 
many instances of a problem which is in \NP{}. Consider the following problem.

\defparproblem{\sc Twin \OCT{} (TOCT)}{Two graphs $G_1$ and $G_2$, terminals $X\subseteq V(G_1)$, $Y\subseteq V(G_2)$, 
a bijection $\Phi$ between $X$ and $Y$, and an integer $k$}{$k$}{Is there a partition of $X$ into three parts 
$(X_1,X_D,X_2)$ such that there is an OCT $Z_1\subseteq V(G_1)\setminus X$ for the graph $G_1-(X_D\cup X_2)$, an OCT $Z_2\subseteq V(G_2)\setminus Y$ for the graph 
$G_2-(\Phi(X_D)\cup \Phi(X_1))$ and $\vert Z_1\cup X_D \cup Z_2\vert\leq k$.}


Clearly the problem \TOCT{} is in \NP. 
Because of Lemma~\ref{22_equivalence1}, for each $V_C\subseteq P_1$ and $V_I\subseteq P_C$ of cardinality at most $4$, we 
construct an instance of \TOCT, of size bounded by a polynomial in $k$, using Lemma~\ref{lemma:oct_strong}. 
After this, we fix a $V_I\subseteq P_C$ and a $V_C\subseteq P_I$, each of cardinality at most $4$. 
Now let $X'=X\cup V_I\cup V_C$.
Note that $X'$ is a \twotwoDeletion{} of $G$ and $(P_I\setminus V_C,P_C\setminus V_I)$ is an \ICpartition{}  
of $G-X'$. The following observation is derived from the fact that $(P_I\setminus V_C,P_C\setminus V_I)$ is an  
\ICpartition{} of $G-X'$ and $V_I\cup V_C\subseteq X'$. 

\begin{observation}\label{solutions_22}
The set $X'$ is an OCT of $G[P_I\cup X']$ and also an OCT of $\overline{G}[P_C\cup X']$. 
\end{observation}

For a particular choice of $V_C\subseteq P_I$ and $V_I\subseteq P_C$ of cardinality at most $4$, 
we construct an instance of \TOCT{} as follows. Let $X'=X\cup V_I \cup C_C$, where $X$ is the 
approximate solution of size bounded by $\Oh(k^{3/2})$. Let $(P_I,P_C)$ be an 
\ICpartition{} of $G-X$. Let $G_1=G[P_I\cup X']$ and $G_2=\overline{G}[P_C\cup X']$. 
By Observation~\ref{solutions_22}, $X'$ is an OCT in graphs $G_1$ and $G_2$. 
Now we apply Lemma~\ref{lemma:oct_strong} and get a set of relevant vertices $Z_1\subseteq V(G_1)\setminus X'$  
of size bounded by $\Oh(k^{9/2})$. Next, we construct a graph $G_1^*$ as follows: 
delete all the vertices $V(G_1)\setminus (X'\cup Z_1)$ from $G_1$. Add two length (three length) path between two 
vertices in $V(G_1^*)$, if there is an even length (odd length) path between the corresponding vertices in $G_1$ 
using only vertices from $V(G)\setminus (X'\cup Z_1)$. 
Similarly, we construct a graph $G_2^*$ from $G_2$. 
Now we output $H=(G_1,G_2,X',X',k)$ as the reduced intance of \TOCT, with the bijection between $X'$ and $X'$
 be the natural identify map.  
Since there are $\Oh(n^4)$ choices for selecting $V_C$ and $V_I$, 
our algorithm will output instances $H_1,H_2,\ldots H_t$ where $t=\Oh(n^4)$ and the size of each $H_i$ is 
bounded by $\Oh(k^{9})$. 

Using Lemmata~\ref{lemma:oct_strong} and ~\ref{22_equivalence1} 
we can prove that in fact the above Turing reduction is correct. 
\begin{lemma}\label{22_equivalence2}
 $(G,k)$ is a \YES{} instance of \twotwopartization{} if and only if there exists $i$ 
 such that $H_i$ is a \YES{} instance of \TOCT.  
%
\end{lemma}
\begin{proof}
 Let $(G,k)$ be a \YES{} instance. 
 Recall that $X$ is an approximate solution and $(P_I,P_C)$ is an 
 \ICpartition{} of $G-X$. 
By Lemma~\ref{22_equivalence1}, we know that there exists $V_C\subseteq P_I$ and $V_I\subseteq P_C$ 
such that the set $X'=X\cup V_I\cup V_C$ can be partitioned into $(X_I'\cup X_D'\cup X_C')$ with the 
following properties. 
\begin{itemize}
 \item there is set $Z_I\subseteq P_I\setminus V_C$ such that $Z_I$ is an OCT of $G[P_I\cup X_I']$. 
 \item there is set $Z_C\subseteq P_C\setminus V_I$ such that $Z_C$ is an OCT of $\overline{G}[P_C\cup X_C']$. 
 \item $\vert Z_I \cup Z_C \cup X_D' \vert \leq k$
\end{itemize}
In our reduction, we have constructed an instance $H_i$ corresponding to the sets $V_C$ and $V_I$. 
That is, $H_i$ is constructed from the graphs $G_1=G[P_I\cup X']$ and $G_2=G[P_C\cup X']$. In the 
construction of $H_i$, we first constructed $G_1^*$ from $G_1$ and $G_2^*$ from $\overline{G_2}$, by finding relevant 
vertices $Z_1$ 
and $Z_2$ in graph $G_1$ and $G_2$ respectively, using Lemma~\ref{lemma:oct_strong}.
Finally we consider the graph $H_i=(G_1^*,G_2^*,X',X',k)$. 
From the construction of 
$G_1^*$ and using Lemma~\ref{lemma:oct_strong}, we know that $G_1^*-(X_D'\cup X_C')$ has an OCT $Z_I^*$ of size at most 
$\vert Z_I\vert$, because $Z_I$ is an OCT in 
$G_1-(X_D'\cup X_C')$. Similarly $G_2^* -(X_D'\cup X_I')$ has an OCT $Z_C^*$ of size at most $\vert Z_C\vert$, 
because $Y_C$ is an OCT in 
$G_2-(X_D'\cup X_I')$. 
This implies that $\vert Y_I^*\cup Y_C^*\cup X_D'\vert \leq k$. 
Thus, $H_i$ is a \YES{} instance of \TOCT. 

In the converse direction, suppose there is an $i$ such that the instance $H_i$
is a \YES{} instance of \TOCT. Note that $H_i$ is constructed for 
a particular $V_C\subseteq P_I$ and $V_I\subseteq P_C$, each of cardinality at most $4$. 
Let $X'=X\cup V_C\cup V_I$. That is the instance $H_i=(G_1^*,G_2^*,X',X',k)$ 
where $G_1^*$ is constructed from $G_1=G[X'\cup P_I]$ and $G_2^*$ is constructed 
from $G_2=\overline{G}[P_C\cup X']$. By our assumption $H_i$ is a \YES{} instance 
of \TOCT. This implies that there is a partition of $X'$ into $(X_1',X_D',X_2')$, 
and that there exists an OCT $Z_I^*$ of $G_1^*-(X_D'\cup X_2')$ and an OCT $Z_C^*$ 
of $G_2*-(X_D'\cup X_1')$ such that $\vert Z_I^*\cup Z_C^*\cup X_D' \vert \leq k$. 
By Lemma~\ref{lemma:oct_strong}, there exists an OCT $Z_I$ of $G_1-(X_D'\cup X_2')$ 
of size at most $\vert Z_I^*\vert $
and an OCT $Z_C$ of $G_2-(X_D'\cup X_1')$ of size at most $\vert Z_C^*\vert $. 
Thus, $\vert Z_I \cup Z_C\cup X_D' \vert \leq k$ and the conditions in the 
Lemma~\ref{22_equivalence1} are met by the partition of $X'$ in to $(X_1',X_D',X_2')$. 
This implies that $(G,k)$ is a \YES{}  instance of \twotwopartization.  \qed
%
%
%
%
\end{proof}
%
The problem \TOCT{} 
is in \NP{} and \twoonepartization{} is \NP-complete. Therefore, by Cook-Levin theorem  each instance $H_i$ 
of \TOCT{} can be reduced to an an instance  of \twotwopartization{} in polynomial time. 
Also note that size of each instance $H_i$ is bounded by $\Oh(k^9)$. 
 Thus we have the following theorem. 
\begin{theorem}
\label{tm:turingkernel}
There is a randomized polynomial Turing kernel for \twotwopartization. 
\end{theorem}
Since there is parameter preserving reduction from \twoonepartization{} and \onetwopartization{} 
to \twotwopartization, the following corollary is derived from Theorem~\ref{tm:turingkernel}. 

\begin{corollary}
 There is a randomized polynomial Turing kernel for \twoonepartization{}  
 and \onetwopartization.
 \end{corollary}

\section{Edge deletion for \abGraph{s}}
%
In this section we show that \twooneedgepartization{} and \onetwoedgepartization{} are in \FPT.

\subsection{\twooneedgepartization} 
In this subsection we show that \twooneedgepartization{} is in \FPT, 
using iterative compression. 
%
For \twooneedgepartization{}, the corresponding compression problem is defined as follows. 
%
%

\defparproblemoutput{{\sc \twooneedgepartization{} Compression}}{A graph $G$ with $V(G)=V\cup \{v\}$, an integer $k$ and an edge set $S'\subseteq E(G-\{v\})$, of size at most $k$,
such that $G[V]- S'$ is a \twooneGraph}{$k$}{A subset $S \subseteq E$  of size at most $k$ such that $G- S$ is a \twooneGraph?}

Like in the case of \twotwopartization, we can show  that \twooneedgepartization{} can be solved,   
by running {\sc \twooneedgepartization{} Compression} at most $\vert V(G) \vert$ times, for an input instance $(G,k)$. 
The following lemma is useful for our purpose. 

\begin{lemma}
 \label{lem:small_clique}
 Let $G$ be a graph on $n$ vertices, $v\in V(G)$ and $\vert E(G- \{v\})\vert \leq k$. Then the number of cliques in $G$ 
 is bounded by $2^{\Oh(\sqrt{k})}n$  and these cliques can be enumerated in time $2^{\Oh(\sqrt{k})}n$. 
\end{lemma}
\begin{proof}
First we bound the size of a maximum clique in $G$ by $\Oh(\sqrt{k})$. Let $\ell$ be the size of a 
maximum clique in $G- \{v\}$. Since the number of edges in $G- \{v\}$ is at most $k$, we have that 
${\ell \choose 2} \leq k$. This implies that $\ell$ is bounded above by $\sqrt{8k}-1$. Since the size of a largest clique in 
$G$ is at most one more than the largest clique in $G- \{v\}$, we have that the size of a maximum clique in $G$ is bounded by $\sqrt{8k}$. It is well known that a graph $H$ on $k$ edges is $\sqrt{2k}$-degenerate (that is every subgraph of $H$ has a vertex of degree at most $\sqrt{2k}$). This implies, that $G$ is $\sqrt{2k}+1$ degenerate. Now, we know from~\cite{EppsteinLS10} that 
$G$ has at most $n3^{\sqrt{2k}+1}$ maximal cliques in $G$ and can be enumerated in time $\Oh(\sqrt{k}n3^{\sqrt{2k}})$. Since every clique in $G$ has size at most $\sqrt{8k}$, given a maximal clique $C$ of $G$, we can generate all the cliques contained in $C$ (by enumerating all subsets of $C$) in time proportional to $2^{\Oh(\sqrt{k})}$. This implies that the number of cliques in $G$ is upper bounded by $2^{\Oh(\sqrt{k})}n$ and it can be enumerated in time $2^{\Oh(\sqrt{k})}n$. \qed
\end{proof}

Next we show that {\sc \twooneedgepartization{} Compression} is in \FPT.

\begin{lemma}
\label{lem:21edgecompression}
{\sc \twooneedgepartization{} Compression} can be solved in time $2^{k+o(k)} \vert V(G)\vert ^{\Oh(1)}$.  
\end{lemma}

\begin{proof}
Let $(G,k,S')$ be the input instance and $|V(G)|=n$. If $G- S'$ is a \twooneGraph, then we return $S'$. 
Otherwise we do the following. 
Let $S$ be a {\em hypothetical solution} for the problem and let $(P_I,P_C)$ be an \ICpartition{} of $G- S$, 
which the algorithm suppose to compute.  Let $G'=G[V]- S'$.  
Since $G'$ is a \twooneGraph, the vertex set $V$ can be partitioned to $I_1$, $I_2$ and $C$ 
such that $G'[I_1]$ and $G'[I_2]$ are graphs with no edges, and $G'[C]$ is a complete graph. 
Since $G'=G[V]- S'$ and $I_1\subseteq V$ and $I_2\subseteq V$ are independent sets in $G'$, 
we have $E(G[I_1])\subseteq S'$ and $E(G[I_2])\subseteq S'$. 
Also, since $\vert S'\vert \leq k$, we have $\vert E(G[I_1])\vert\leq k$ and $\vert E(G[I_2])\vert \leq k$. 
Now consider the partition of the vertex set of $G$, $V\cup \{v\}$, into three parts $I_1\cup \{v\}$, $I_2$ and $C$. 
Recall that $(P_I,P_C)$ is an \ICpartition{} of our hypothetical solution $S$. 
Our algorithm guesses the sets of vertices $A=(I_1\cup \{v\})\cap P_C$ and $B=I_2\cap P_C$.
Since the partition $P_C$ should be a clique, $A\cup B$ is a clique. Thus, guessing 
the vertex sets $A$ and $B$ from $I_1\cup \{v\}$ and $I_2$ respectively is equal to guessing 
two cliques from $G[I_1\cup \{v\}]$ and $G[I_2]$ such that they together form a clique in $G$. 
By Lemma~\ref{lem:small_clique}, the number of cliques in $G[I_1\cup \{v\}]$ and $G[I_2]$ 
is bounded by $2^{\Oh(\sqrt{k})}n$ and these clique can be enumerated in time $2^{\Oh(\sqrt{k})}n$. 

After guessing $A$ and $B$, we know that in our hypothetical \ICpartition{} $(P_I,P_C)$, $A\cup B\subseteq P_C$ 
and $(I_1\cup I_2 \cup \{v\})- (A\cup B) \subseteq P_I$. Let 
$C'=\{u\in C~\vert ~ A\cup B \subseteq N(u)\}$. 
The following claim implies that we can set $P_C=A\cup B\cup C'$. 
\begin{claim}
\label{claim21edge}
 If there is a subset $S_1\subseteq E(G)$ and a partition $(P_I',P_C')$ of $V(G)$ such that 
 $(i)\; G- S_1$ is a \twooneGraph, $(ii)\;(P_I',P_C')$ is an \ICpartition{} of $G-
S_1$, $(iii)\, (I_1\cup I_2 \cup \{v\})- (A\cup B) \subseteq P_I'$ and $(iv)\; A\cup B\subseteq P_C'$, then 
 there is a subset $S_2\subseteq E(G)$ and a partition $(P_I'',P_C'')$ of $V(G)$ such that $(i)\; \vert S_2\vert \leq \vert S_1\vert $, $(ii)\; G- S_2$ is a \twooneGraph,
$(iii)\; (P_I'',P_C'')$ is an \ICpartition{}  of $G- S_2$,   
 $(iv)\;(I_1\cup I_2 \cup \{v\})- (A\cup B) \subseteq P_I''$ and $(v)\; A\cup B\cup C' = P_C''$.  
\end{claim}
\begin{proof}
We have given a set $S_1$ and an \ICpartition{} $(P_I',P_C')$ of $G- S_1$ with properties mentioned in the statement of the  claim. 
Since $(P_I',P_C')$ is an \ICpartition{} of $G- S_1$, $S_1$ is an edge OCT of $G[P_I']$. Since $(I_1\cup I_2 \cup \{v\})- (A\cup B) \subseteq P_I'$, 
we have that $P_C'- (A\cup B)\subseteq C$. Also since $P_C'$ is a clique and $A\cup B\subseteq P_C'$, we have that $P_C'- (A\cup B)\subseteq C'$. 
Now consider the partition $(P_I'',P_C'')$ of $V(G)$, where $P_C''=A\cup B \cup C'$ and $P_I''=V(G)- P_C''$. Note that 
$P_C''$ is a clique and $P_I''\subseteq P_I'$. This implies that the edge set $S_2=S_1 \cap E(G[P_I''])$ is 
an \OCET{} set of $G[P_I'']$. Hence the set $S_2$ and the partition $(P_I'',P_C'')$ are the required set and  the partition, respectively,  in the claim. \qed
\end{proof}
Claim~\ref{claim21edge} implies that for the correct guess of $A$ and $B$, we can set $P_C=A\cup B\cup C'$. 
This in turn implies that the problem is now reduced to delete as few edges as possible to make the graph $G- (A\cup B\cup C')$ bipartite and 
this is nothing but the \OCET{} problem on $(G- (A\cup B\cup C'),k)$. The problem $\OCET{}$ can be solved in time $2^k n^{\Oh(1)}$, where 
$n$ is the number of vertices in the input graph~\cite{GuoGHNW06}. Since there are  $2^{\Oh(\sqrt{k})}n^2$ choices for guessing $A$ and $B$, the total running time of the algorithm 
is bounded by $2^{k+o(k)} n^{\Oh(1)}$. \qed
\end{proof}

Thus by using  Lemma~\ref{lem:21edgecompression}, we can get the following theorem. 
\begin{theorem}
 \label{thm:21edge}
\twooneedgepartization{} can be solved in time $2^{k+o(k)} \vert V(G)\vert ^{\Oh(1)}$.  
\end{theorem}

\subsection{\onetwoedgepartization}
In this subsection we show that \onetwoedgepartization{} is in \FPT. Again we use the  
iterative compression technique to solve the problem. For our algorithm, we need  
an algorithm for a version of \OCT.
Let $\G$ be an hereditary graph class (hereditary means that if $G \in \G$, then every induced subgraph of $G$ is in $\G$ as well) 
and $\G$ is decidable. Then the problem {\sc $\G$-Weighted Bipartition} is defined as follows. 

\defparproblemoutput{\sc $\G$-Weighted Bipartition}{A graph $G$, $w:V(G)\rightarrow {\mathbb N}^+$ and integers $k$ and $W$}{$k+W$}
{An OCT $O$ of $G$, of size at most $k$ such that $w(O)\leq W$ and  $G[O]\in \G$}

Marx et al.~\cite{MarxOR13} showed that the unweighted version of the problem, named, {\sc $\G$-Bipartition} 
can be solved in \FPT{} time. The proof by Marx et al., constructs an ``equivalent graph'' with 
treewidth bounded by a function of $k$. The problem is then solved in the equivalent graph, using 
Courcelle's theorem~\cite{Courcelle90} by expressing the problem as an MSO predicate. Since we can express whether the weight of a subset of vertices is at most $W$
using an MSO predicate of length bounded by a function of $W$, the following theorem follows from the results of Marx et al.~\cite{MarxOR13}.  
\begin{theorem}
\label{thm:specialOCT}
If $\G$ is hereditary and decidable, then 
{\sc $\G$-Weighted Bipartition} is in \FPT.  
\end{theorem}

Now we are ready to define compression version of the problem  \onetwoedgepartization{} and 
prove that it is in \FPT, which in turn will imply that non-compression version of the problem is in \FPT.

\defparproblemoutput{{\sc \onetwoedgepartization{} Compression}}{A Graph $G$ with $V(G)=V\cup\{v\}$, an integer $k$ and an edge set $S'\subseteq E(G-v)$, of size at most $k$,
such that $G[V]- S'$ is a \onetwoGraph}{$k$}{A subset $S \subseteq E$  of size at most $k$ such that $G- S$ is a \onetwoGraph?} 


\begin{lemma}
\label{lem:12edgecompression}
{\sc \onetwoedgepartization{} Compression} is in \FPT.   
\end{lemma}
\begin{proof}
Let $(G,k,S')$ be the input instance and $|V(G)|=n$. If $G- S'$ is a \onetwoGraph, then we return $S'$ 
as the output.  Otherwise we do the following. 
Let $S$ be a {\em hypothetical solution} for the problem and let $(P_I,P_C)$ be an \ICpartition{} of $G- S$. 
Let $G'=G[V]- S'$. Since $G'$ is a \onetwoGraph, the vertex set $V$ can be partitioned to $I$, $C_1$ and $C_2$ 
such that $(i)$ $G'[I]$ is a graph with no edges, and $(ii)$ $G'[C_1]$ and $G'[C_2]$ are cliques. 
Since $G'=G[V]- S'$ and $I\subseteq V$ is independent sets in $G'$, we have that $E(G[I])\subseteq S'$. 
Also since $\vert S'\vert \leq k$, we have that $\vert E(G[I])\vert\leq k$.  
Now consider the partition of the vertex set of $G$, $V(G)$,  into three parts $I$, $C_1\cup \{v\}$ and $C_2$. 
Recall that $(P_I,P_C)$ is an \ICpartition{} of our hypothetical solution $S$. 
Now our algorithm guesses the set of vertices $A=I \cap P_C$. 
Since $P_C$ should be a complement of a bipartite graph, 
$A$ should also be a complement of a bipartite graph. Hence our algorithm guesses 
two cliques $K_1$ and $K_2$ from $G[I]$ and assumes that $A=K_1\cup K_2$ will be 
part of $P_C$. By Lemma~\ref{lem:small_clique}, we have that the number of cliques in $G[I]$  
is bounded by $2^{\Oh(\sqrt{k})}n$  and these cliques can be enumerated in time $2^{\Oh(\sqrt{k})}n$.  
After guessing $A$, we know that in our hypothetical \ICpartition{} $(P_I,P_C)$, $I- A \subseteq P_I$ and $A\subseteq P_C$.  
Now consider the partition $(P_I',P_C')$ of $V(G)$, where $P_I'=I- A$ and $P_C'=A\cup C_1 \cup C_2 \cup \{v\}$. 

Now to solve the problem it is enough to find out a subset $U\subseteq  C_1 \cup C_2 \cup \{v\}$ 
such that $U$ is an OCT of the complement graph of $G[P_C']$ and $\vert E(G[P_I'\cup U])\vert \leq k$. 
This can be encoded as a {\sc $\G$-Weighted Bipartition} 
problem.
Since $U\subseteq C_1\cup C_2 \cup \{v\}$ and $C_1$ and $C_2$ are cliques, 
the cardinality of the set $U$ will be bounded by $\Oh(\sqrt k)$.
The edges that contribute to $E(G[P_I'\cup U])$ are of three types--$(i)$ edges within $G[P_I']$, 
$(ii)$ edges in $G[U]$ and $(iii)$ edges between $U$ and $P_I'$ in $G$.  
Let $k_1=\vert E(G[P_I'])\vert$. To encode the edges between $U$ and $P_I'$ we introduce 
a weight function $w$ on $P_C'$. For each $u\in P_C'$, $w(u)=\vert N_G(u)\cap P_I'\vert$. 
Since we have fixed $P_I'$ to be a subset of $P_I$, we need to include the set 
of edges in $E(G[P_I'])$ (type $(i)$) in the solution of the problem. The rest of the edges in the solution come 
from type $(ii)$ or type $(iii)$. Let $k_1= \vert E(G[P_I'])\vert $, $k_2=E(G[U])$ and $k_3$ 
be the number edges between $U$ and $P_I'$. 
So $U$ is an OCT in the complement of $G[P_C']$, of weight at most $k_3$ and 
number of edges in $G[U]$ is bounded by $k_2$. 
Now our algorithm guesses the number 
of edges of type $(ii)$ to be $k_2$ and type $(iii)$ to be $k_3$. 
Let $\G_{k_2}$ be the class of graphs such that the number of edges in it is bounded by $k_2$. 
The class $\G_{k_2}$ is hereditary. To solve our problem it is 
enough to solve {\sc $\G_{k_2}$-Weighted Bipartition} on the complement of the graph $G[P_C']$ with weight function $w$. 
This completes the proof of the lemma. \qed
\end{proof}

Thus by using  Lemma~\ref{lem:12edgecompression}, we can get the following theorem. 
\begin{theorem}
 \label{thm:12edge}
\onetwoedgepartization{} is in \FPT.  
\end{theorem}

\bibliographystyle{abbrv}
\bibliography{references,ref}

\newpage
\appendix
\section{Proof of Lemma~\ref{lemma:oct_strong}}
In this subsection we give an outline of proof of Lemma~\ref{lemma:oct_strong}. 
The results in \cite{ReedSV04} show that the \OCT{} problem is 
equivalent to many instances of minimum separator problem in an auxiliary graph of the input graph and this is the main idea used to show 
\OCT{} is in \FPT. 
Let $G$ be a graph and $X$ is an OCT of $G$.
Note that $G-X$ is a bipartite graph. 
Without loss of generality we may assume that $X$ is independent, otherwise we can subdivide the edge with in $X$ 
and still $X$ be an OCT of $G$. 
Let $S_1\uplus S_2$ be a bipartition of $G-X$.
The auxiliary graph $G'$ of $G$ is constructed in~\cite{ReedSV04} is as follows. 
the vertex set of $G$, $V(G') = V(G)\setminus X \cup \{x_1,x_2 | x\in X\}$ and $E(G-X)\subseteq E'$. Additionally, we
add edges between $x_1$ and neighbors of $x$ in $S_2$, and between $x_2$ and
neighbors of $x$ in $S_1$. 
Given $U\subseteq X$, a valid partition of $X'(U)$ is pair $(S,T)$ which satisfies the following properties. 
\begin{enumerate}
\item $S\uplus T = X'(U)$; 
\item For every $x\in U$, $|\{x_1,x_2\}\cap S|=|\{x_1,x_2\}\cap T|=1$;
\item For every $x\in X\setminus U$, $|\{x_1,x_2\}\cap S|=|\{x_1,x_2\}\cap T|=0$.
\end{enumerate}
Now the following lemma follows from the results of Reed et al.~\cite{ReedSV04} (for more details see Lemma 3.2 and Lemma 3.3 in~\cite{LokshtanovSS09})
\begin{lemma}
 \label{lemma:OCTcut}
 Let $G$ be a graph and $X$ is an OCT of $G$. Let $G'$ is the auxiliary graph constructed from $G$. 
 For any $Y\subseteq X$, $O$ is an OCT  of $G-Y$ 
 not containing $X$ 
if and only if there is a valid 
 partition $(S,T)$ of $X'(X\setminus Y)$ such that  $O$ is an $(S,T)$-vertex cut in $G'-X'(Y)$. 
\end{lemma}
The following {\em cut covering lemma} is proved in~\cite{DBLP:conf/focs/KratschW12}. 
\begin{lemma}
 \label{lemma:cutcovering}
Let $G$ be a graph, and $X\subseteq V(G)$ a set of terminals. We can identify,
in randomized polynomial time, a set $Z$ of 
$O(|X|^3)$ vertices such that for any $S, T, R \subseteq X$, a minimum $(S,T)$-vertex cut in $G-R$ is contained in $Z$.
 \end{lemma} 
Now the proof of Lemma~\ref{lemma:oct_strong} follows directly from Lemma~\ref{lemma:OCTcut} and Lemma~\ref{lemma:cutcovering}.


\end{document}